\documentclass[a4paper]{llncs}
\usepackage{longtable}
\usepackage[table]{xcolor}
\usepackage{stmaryrd}
\usepackage{amsmath}
\usepackage{mathtools}
\usepackage{setspace}
\usepackage{cite}
\usepackage{color}

\DeclareMathAlphabet\mathbfcal{OMS}{cmsy}{b}{n}

\newcommand{\defAs}{\mathrel{\makebox{\:= \hspace{-.2cm} \raisebox{-0.5 ex}[0 ex][0 ex]{\tiny Def}}}}
\newcommand{\dlss}{\mathcal{DL}\langle 4LQS^R\rangle(\D)}
\newcommand{\shdlss}{\mathcal{DL}_{\D}^{4}}
\newcommand{\sroiqd}{\ensuremath{\mathcal{SROIQ}(\D)}}
\newcommand{\flqsr}{\ensuremath{4LQS^R}}
\newcommand{\I}{\mathbf{I}}
\newcommand{\C}{\mathbf{C}}
\newcommand{\D}{\mathbf{D}}

\newcommand{\Ra}{\mathbf{R_A}}
\newcommand{\Rd}{\mathbf{R_D}}
\newcommand{\sym}{\mathsf{Sym}}
\newcommand{\asym}{\mathsf{Asym}}
\newcommand{\refl}{\mathsf{Ref}}
\newcommand{\irref}{\mathsf{Irref}}
\newcommand{\tra}{\mathsf{Tra}}
\newcommand{\fun}{\mathsf{Fun}}
\newcommand{\ck}{\mathsf{cpt}_\mathcal{K}}
\newcommand{\ark}{\mathsf{arl}_\mathcal{K}}
\newcommand{\crk}{\mathsf{crl}_\mathcal{K}}
\newcommand{\ik}{\mathsf{ind}_\mathcal{K}}

\newcommand{\bfk}{\mathsf{bf}_\mathcal{K}^{\D}}
\newcommand{\vipcomment}[1]{}
\newcommand{\pow}{\textit{pow}}

\title{Web ontology representation and reasoning via fragments of set theory}

\author{Domenico Cantone \and  Cristiano Longo \and Marianna Nicolosi-Asmundo \and Daniele Francesco Santamaria}

\institute{
University of Catania, Dept. of Mathematics and Computer Science\\
~email:~\texttt{\{cantone,longo,nicolosi\}@dmi.unict.it, daniele.f.santamaria@gmail.com}
}

\begin{document}
\maketitle

\begin{abstract}
In this paper we use results from Computable Set Theory as a means to represent and reason about description logics and rule languages for the semantic web.

Specifically, we introduce the description logic $\dlss$--admitting features such as min/max cardinality constructs on the left-hand/right-hand side of inclusion axioms, role chain axioms, and datatypes--which turns out to be quite expressive if compared with \sroiqd, the description logic underpinning the Web Ontology Language OWL.
Then we show that the consistency problem for $\dlss$-knowledge bases is decidable by reducing it, through a suitable translation process, to the satisfiability problem of the stratified fragment $4LQS^R$ of set theory, involving variables of four sorts and a restricted form of quantification.
We prove also that, under suitable not very restrictive constraints, the consistency problem for $\dlss$-knowledge bases is \textbf{NP}-complete.
Finally, we provide a $4LQS^R$-translation of rules belonging to the Semantic Web Rule Language (SWRL).

%
\end{abstract}

%
%
%
%
%

\section{Introduction}

Computable Set Theory is a research field started in the late seventies with the purpose of studying the decidability of the satisfiability problem for fragments of set theory. The most efficient decision procedures designed in this area have been implemented within the reasoner $\mathsf{{\AE}tnanova/Referee}$ \cite{SchwCanOmoPol11} and constitute its inferential core.
A wide collection of decidability results obtained up to 2001 can be found in the monographs \cite{CaFeOm90,CaOmPo01}.

Most of the decidability results and applications in computable set theory concern one-sorted multi-level syllogistics, namely collections of formulae admitting variables of one sort only, which range over the Von Neumann universe of sets. Only a few stratified syllogistics, where variables of multiple sorts are allowed, have been investigated, despite the fact that in many fields of computer science and mathematics one often has to deal with multi-sorted languages.
For instance, in \emph{Description Logics} one has to consider entities of different types, namely \emph{individual elements}, \emph{concepts}, namely sets of individuals, and \emph{roles}, namely binary relations over elements.


\newcommand{\mlsscart}{\ensuremath{\mathsf{MLSS}_{2,m}^{\times}}}
\newcommand{\dlmlsscart}{\ensuremath{\mathcal{DL\langle}\mathsf{MLSS}_{2,m}^{\times}\mathcal{\rangle}}}
\newcommand{\dlForallpizerotwo}{\ensuremath{\mathcal{DL\langle}\Forallpizerotwo\mathcal{\rangle}}}
\newcommand{\Forallpizerotwo}{\ensuremath{\mathbf{\forall_{0,2}^{\pi}}}}

Recently, one-sorted multi-level fragments of set theory allowing one to express constructs related to \emph{multi-valued maps} have been studied (see \cite{CanLonNic2010, CanLonNic11, CanLon2014}) and applied in the realm of knowledge representation. In \cite{CanLonPis2010}, for instance, an expressive description logic, called $\dlmlsscart$, has been introduced and the consistency problem for $\dlmlsscart$-knowledge bases has been proved $\mathbf{NP}$-complete. $\dlmlsscart$ has been extended with additional description logic constructs and SWRL rules in \cite{CanLonNic11}, proving that the decision problem for the resulting description logic, called $\dlForallpizerotwo$, is still $\mathbf{NP}$-complete under some conditions. Finally, in \cite{CanLon2014} $\dlForallpizerotwo$ has been extended with some \emph{metamodelling} features.
However, none of the above-mentioned description logics provides any functionality to deal with
datatypes, a simple form of concrete domains that are relevant in real-world applications.


In this paper we introduce an expressive description logic, $\dlss$ (more simply referred to as $\shdlss$ in the rest of the paper), that can be represented in the decidable four-level stratified fragment of set theory \flqsr. The logic $\shdlss$ supports da\-ta\-ty\-pes, and admits concept constructs such as full negation, union and intersection of concepts, concept domain and range, existential quantification and min cardinality on the left-hand side of inclusion axioms, universal quantification and max cardinality on the right-hand side of inclusion axioms. It also supports role constructs such as role chains on the left hand side of inclusion axioms, union, intersection, and complement of roles, and properties on roles such as transitivity, symmetry, reflexivity, and irreflexivity.

We shall prove that the consistency problem for $\shdlss$-knowledge bases is decidable via a reduction to the satisfiability problem for formulae of \flqsr. The latter problem was proved decidable in \cite{CanNic2013}.  We shall also show that the consistency problem for $\shdlss$-knowledge bases involving only suitably constrained $\shdlss$-formulae is $\mathbf{NP}$-complete. Such restrictions are not very limitative: in fact, it turns out that the constrained logic allows one to represent real world ontologies such as \textsf{Ontoceramic}, designed for ancient ceramic cataloguing in collaboration with archaeological experts (see \cite{caa2015:ontoceramic,santaLAP}).

The logic $\shdlss$ is not an extension of \sroiqd, the description logic upon which the W3C standard OWL 2 DL is based, as it admits existential (resp., universal) quantification only on the left-hand (resp., right-hand) side of inclusion axioms. However, $\shdlss$ supports chain axioms that are more liberal than the ones supported by \sroiqd, as they can involve roles that are not subject to any regularity restriction. Moreover, Boolean combination of roles is supported even on the right-hand side of chain axioms. The latter fact is particularly relevant to the problem of expressing rules in OWL. We will
briefly illustrate how \flqsr\space can be used to express SWRL rules in Section \ref{rules}.


The paper is organized as follows. In Section \ref{prelim} we review the syntax and semantics of the set-theoretic fragment \flqsr\ and of the logic \sroiqd. Then, in Section \ref{dlss}, we present the description logic $\shdlss$ 
and prove that the decidability of the consistency problem for $\shdlss$-knowledge bases can be reduced to the satisfiability problem for \flqsr-formulae. In particular, in Section \ref{rules} we show that SWRL rules can be represented within the \flqsr-fragment. 
Finally, in Section \ref{conclusions} we draw our conclusions and give some hints to future work.

\section{Preliminaries}\label{prelim}
In this section we introduce concepts and notions that will be used in the paper.

\subsection{The set-theoretic fragment \flqsr} \label{4LQS}
In order to define the fragment \flqsr, it is convenient to first introduce the syntax and semantics of a more general four-level quantified language, denoted $4LQS$.
Then 
we provide some restrictions on quantified formulae of $4LQS$ that characterize \flqsr.
%
%
We recall that the satisfiability problem for \flqsr\ has been proved decidable in \cite{CanNic2013}.

$4LQS$ involves the four collections of variables $\mathcal{V}_0$, $\mathcal{V}_1$, $\mathcal{V}_2$, $\mathcal{V}_3$, where:

\smallskip
{-  $\mathcal{V}_0$ contains variables of sort $0$, denoted by $x, y, z, ...$; }

{- $\mathcal{V}_1 $ contains variables of sort $1$, denoted by $X^1,Y^1,Z^1,... $; }

{- $\mathcal{V}_2$  contains variables of sort $2$, denoted by $X^2,Y^2,Z^2,... $;}

{- $\mathcal{V}_3$ contains variables of sort $3$, denoted by $X^3,Y^3,Z^3,...$ .}

\medskip
\noindent
In addition to variables, $4LQS$ involves also \emph{pair terms} of the form $\langle x,y \rangle$, for $ x,y \in \mathcal{V}_0$.
\emph{$4LQS$-quantifier-free atomic formulae} are classified as:
\begin{itemize}
\item[-] {level 0: $x=y$, $x \in X^1$, $\langle x,y \rangle = X^2$, $\langle x,y \rangle \in X^3$, where $x,y \in \mathcal{V}_0$, $\langle x,y \rangle$ is a pair term, $X^1 \in\mathcal{V}_1$, $X^2\in \mathcal{V}_2$, $X^3$ in $\mathcal{V}_3$;}
\item[-] {level 1: $X^1=Y^1$, $X^1 \in X^2$, with $X^1,Y^1 \in \mathcal{V}_1$, $X^2$ in $\mathcal{V}_2$;}
\item[-] {level 2: $X^2=Y^2$, $X^2 \in X^3$, with $X^2, Y^2 \in \mathcal{V}_2$, $X^3$ in $\mathcal{V}_3$.}
\end{itemize}

\noindent $4LQS$ \emph{purely universal formulae} are classified as:

\begin{itemize}
\item[-] { level 1: $(\forall z_1)...(\forall z_n) \varphi _0$, where $z_1,..,z_n$  $\in \mathcal{V}_0$ and $\varphi _0$ is any propositional combination of quantifier-free atomic formulae of level 0;}
\item[-] { level 2: $(\forall Z^1_1)...(\forall Z^1_m) \varphi _1$, where $Z^1_1,..,Z^1_m $  $\in \mathcal{V}_1$ and $\varphi _1$ is any propositional combination of quantifier-free atomic formulae of levels 0 and 1 and of purely universal formulae of level 1;}
\item[-] {level 3: $(\forall Z^2_1)...(\forall Z^2_p) \varphi _2$, where $Z^2_1,..,Z^2_p $  $\in \mathcal{V}_2$ and $\varphi _2$ is any propositional combination of quantifier-free atomic formulae and of purely universal formulae of levels 1 and 2.}
\end{itemize}

\noindent
$4LQS$-formulae are all the propositional combinations of quantifier-free atomic formulae of levels 0, 1, 2 and of purely universal formulae of levels 1, 2, 3.

Let $\varphi$ be a $4LQS$-formula. Without loss of generality, we can assume that $\varphi$ contains only $\neg$, $\wedge$, $\vee$ as propositional connectives. Further, let $S_{\varphi}$ be the syntax tree for a $4LQS$-formula $\varphi$,\footnote{The notion of syntax tree for $4LQS$-formulae is similar to the notion of syntax tree for formulae of first-order logic. A precise definition of the latter can be found in \cite{DeJo90}.} and let $\nu$ be a node of $S_{\varphi}$. We say that a $4LQS$-formula $\psi$ occurs within $\varphi$ at position $\nu$ if the subtree of $S_{\varphi}$ rooted at $\nu$
is identical to $S_{\psi}$. In this case we refer to $\nu$ as an occurrence of $\psi$ in $\varphi$ and to the path from the root
of $S_{\varphi}$ to $\nu$ as its occurrence path. An occurrence of $\psi$ within $\varphi$ is
\emph{positive} if its occurrence path deprived by its last node contains an even number of nodes labelled by a $4LQS$-formula of type $\neg \chi$. Otherwise, the occurrence is said to be \emph{negative}.

A $4LQS$-\emph{interpretation} is a pair $\mathbfcal{M}=(D,M)$ where $D$ is any non-empty collection of objects (called domain or universe of $\mathbfcal{M}$) and $M$ is an assignment over variables in $\mathcal{V}_0$, $\mathcal{V}_1$, $\mathcal{V}_2$, $\mathcal{V}_3$ such that

\smallskip
{- $Mx \in D$, for each $ x \in \mathcal{V}_0$;}

{- $MX^1 \in \pow(D)$, for each $X^1 \in \mathcal{V}_1$;}

{- $ MX^2 \in \pow(\pow(D))$, for each $X^2 \in \mathcal{V}_2$; }

{- $MX^3 \in \pow(\pow(\pow(D)))$, for each $X^3 \in \mathcal{V}_3$

(we recall that $\pow(s)$ denotes the powerset of $s$).}

\smallskip
\noindent
We assume that pair terms are interpreted \emph{\`a la} Kuratowski, and therefore we put $ M \langle x,y \rangle =_{Def} \{ \{ Mx \},\{ Mx,My \} \}$. The presence of a pairing operator in the language is very useful for the set theoretic representation of the logic $\shdlss$ and of SWRL rules introduced in Sections \ref{dlss} and \ref{rules}, respectively. Moreover, even though several pairing operators are available (see \cite{FormisanoOP04}), encoding ordered pairs \`a la Kuratowski turns out to be quite straightforward, at least for our purposes.


Next, let

\smallskip
{- $\mathbfcal{M}=(D,M)$ be a $4LQS$-interpretation,}

{- $x_1,....x_n \in \mathcal{V}_0$,~~ $X^1 _1, ... X^1 _m \in \mathcal{V}_1$,~~ $X^2 _1, ... X^2 _p \in \mathcal{V}_2$,}

{- $u_1, ... u_n \in D$,~~ $U^1 _1, ... U^1 _m \in \pow(D)$,~~ $U^2 _1, ... U^2 _p \in \pow(\pow(D))$.}

\smallskip
\noindent
By  $\mathbfcal{M}  [ x_1 / u _1, ..., x_n  / u_n, X^1 _1 / U^1 _1, ... X^1_m / U^1_m, X^2_1 / U^2_1, ... X^2_p / U^2_p  ] $, we denote the interpretation $\mathbfcal{M}'=(D,M')$ such that $M'x_i =u_i$, for $i=1,...,n$, $M'X^1_j = U^1 _j$, for $j=1,...,m$, $M'X^2_k = U^2_k$, for $k=1,...,p$, and which otherwise coincides with $M$ on all remaining variables. Let $\varphi$ be a $4LQS$-formula and let $\mathbfcal{M} =(D,M)$ be a $4LQS$-interpretation. The notion of satisfiability of $\varphi$ by $\mathbfcal{M} $ (denoted by  $ \mathbfcal{M} \models \varphi$) is defined inductively over the structure of $\varphi$. Quantifier-free atomic formulae are evaluated in a standard way according to the usual meaning of the predicates `$\in$'
and `$=$', and purely universal formulae are evaluated as follows:
\begin{itemize}
\item[-]{ $\mathbfcal{M}  \models (\forall z_1) ... (\forall z_n) \varphi _0$ iff $\mathbfcal{M}   [ z_1 / u _1,...,z_n  / u_n] \models \varphi_0$, for all $u_1, ... u_n \in D ; $  }
\item[-]{ $\mathbfcal{M}  \models (\forall Z^1_1) ... (\forall Z^1_m) \varphi _1$ iff $\mathbfcal{M}   [ Z^1_1 / U^1 _1,...,Z^1_n  / U^1_n] \models \varphi_1$, for all $U^1_1, ... U^1_m \in \pow(D) ; $  }
\item[-]{ $\mathbfcal{M}  \models (\forall Z^2_1) ... (\forall Z^2_m) \varphi _2$ iff $\mathbfcal{M}   [ Z^2_1 / U^2 _1,...,Z^2_n  / U^2_n] \models \varphi_2$, for all $U^2_1, ... U^2_m \in \pow(\pow(D)).$  }
\end{itemize}
Finally, compound formulae are interpreted according to the standard rules of propositional logic. If $\mathbfcal{M} \models \varphi$, then $\mathbfcal{M} $ is said to be a $4LQS$-model for $\varphi$. A $4LQS$-formula is said to be \emph{satisfiable} if it has a $4LQS$-model. A $4LQS$-formula is \emph{valid} if it is satisfied by all $4LQS$-interpretations.

Next we present the fragment \flqsr\ of $4LQS$ of our interest, namely the collection of the formulae $\psi$ of $4LQS$ fulfilling the restrictions:
\begin{enumerate}
\item{for every purely universal formula $(\forall Z^1_1),...,(\forall Z^1_m) \varphi_1$ of level 2 occurring in $\psi$ and every purely universal formula $(\forall z_1),...,(\forall z_n) \varphi_0$ of level 1 occurring negatively in $\varphi_1$,
    the condition
\[ {\small
\neg \varphi_0 \rightarrow \overset{n}{ \underset {i=1} \bigwedge} \; \overset {m} { \underset {j=1 }\bigwedge} z_i \in Z^1_j }\]
is a valid $4LQS$-formula (in this case we say that $(\forall z_1),...,(\forall z_n) \varphi_0$ is \emph{linked to the variables} $Z^1_1,...,Z^1_m$);
}

\item{ for every purely universal formula  $(\forall Z^2_1),...,(\forall Z^2_p) \varphi_2$  of level 3 in $\psi$:
\begin{itemize}
\item[-]{every purely universal formula of level 1 occurring negatively in $\varphi_2$ and not occurring in a purely universal formula of level 2 is only allowed to be of the form
\[
(\forall z_1),..., (\forall z_n) \neg( \overset {n}{ \underset {i=1} \bigwedge} \; \overset {n} { \underset {j=1}\bigwedge} \langle z_i,z_j \rangle=Y^2_{ij}),\]
with $Y^2_{ij} \in \mathcal{V}^2$, for $i,j=1,...,n$;}
\item[-]{purely universal formulae $(\forall Z^1_1),...,(\forall Z^1_m) \varphi_1$ of level 2 may occur only positively in $\varphi_2$.}
\end{itemize}
}
\end{enumerate}

Restriction 1 has been introduced for technical reasons concerning the decidability of the satisfiability problem for the fragment. In fact it guarantees that satisfiability is preserved in a suitable finite submodel of $\psi$.  Restriction  2 allows one to express binary relations and several operations on them while keeping simple, at the same time, the decision procedure (for space reasons details are not included here but can be found in \cite{CanNic2013}).

We observe that the semantics of \flqsr\space plainly coincides with that of $4LQS$.

In the \flqsr-fragment one can express several set-theoretic constructs such as a restricted variant of the set former, which in
turn allows one to express other significant set operators such as binary union, intersection, set difference,
the singleton operator, the powerset operator, etc. Within the fragment \flqsr, it is also possible to define binary relations over elements of a domain
together with conditions on them (i.e., reflexivity, transitivity, weak connectedness, irreflexivity, intransitivity) which characterize accessibility relations of well-known modal logics. In particular, the normal modal logic $\mathsf{K45}$ can be translated in the \flqsr-fragment. Again, the interested reader is referred to \cite{CanNic2013} for details.

\subsection{Description Logics}
Description Logics (DL) are a family of formalisms widely used in the field of Knowledge Representation to model application domains and to reason on them \cite{Baader03descriptionlogics}.
DL knowledge bases describe models that are based on individual elements (or, more simply, individuals), classes whose elements are individual names, and binary relationships between individuals. These three types of semantic entities are syntactically denoted by means of individual names, concept names, and role names. In addition, DL provide operators for combining concept and role names into complex concept and role expressions.
One of the leading application domains for DL is the semantic web. In fact, the most recently developed semantic web language, namely OWL 2, is based on a very expressive description logic with datatypes $\D$, called \sroiqd. Extensions of DL with datatypes have been studied and analyzed in \cite{Horrocks:2001a, Motik2008}.

\vipcomment{Datatype reasoning can be performed using an external datatype checking procedure. Standard DL tableau calculi can then be extended to handle datatypes by invoking the datatype checker as an oracle.  In particular, in \cite{Motik2008} it is shown that datatype checking in OWL 2 is NP-hard in the general case, but may become trivial in many (hopefully typical) cases. \marginpar{La parte in rosso potrebbe essere rimossa.} The authors also argue that certain datatypes listed as normative in the current OWL 2 Working Draft may be unsuitable from both a modeling and implementation perspective. They also present a modular datatype checking algorithm that can support any datatype for which it is possible to implement a small set of basic operations that we call a datatype handler and discuss how to implement datatype handlers for number and string datatypes}

The logic \sroiqd\space is briefly introduced in the next section (the interested reader is referred to \cite{Horrocks2006} for details).

\subsubsection{The description logic \sroiqd.}\label{sroiqd}
Let $\D = (N_{D}, N_{C},N_{F},\cdot^{\D})$ be a \emph{datatype map} in the sense of \cite{Motik2008}, where  $N_{D}$ is a finite set of datatypes, $N_{C}$ is a function assigning a set of constants $N_{C}(d)$ to each datatype $d \in N_{D}$, $N_{F}$ is a function assigning a set of facets $N_{F}(d)$ to each $d \in N_{D}$, and $\cdot^{\D}$ is a function assigning a datatype interpretation $d^{\D}$ to each datatype $d \in N_{D}$, a facet interpretation $f^{\D} \subseteq d^{\D}$ to each facet $f \in N_{F}(d)$, and a data value $e_{d}^{\D} \in d^{\D}$ to every constant $e_{d} \in N_{C}(d)$.  We shall assume that the interpretations of the datatypes in $N_{D}$ are nonempty pairwise disjoint sets.

A \emph{facet expression} for a datatype $d \in N_{D}$ is a formula $\psi_d$ constructed from the elements of $N_{F}(d) \cup \{\top_{d},\bot_{d}\}$ by applying a finite number of times the connectives $\neg$, $\wedge$, and $\vee$. The function  $\cdot^{\D}$ is extended to facet expressions for $d  \in N_{D}$ by putting $\top_{d}^{\D} = d^{\D}$, $\bot_{d}^{\D} = \emptyset$, $(\neg f)^{\D} = d^{\D} \setminus f^{\D}$, $(f_1 \wedge f_2)^{\D} = f_1^{\D} \cap f_2^{\D}$, and $(f_1 \vee f_2)^{\D} = f_1^{\D} \cup f_2^{\D}$, for $f, f_1,f_2 \in N_{F}(d)$.

A \emph{data range} $dr$ for $\D$ is either a datatype $d \in N_{D}$, or a finite enumeration of datatype constants $\{e_{d_1},\ldots,e_{d_n}\}$, with $e_{d_i} \in N_{C}(d_i)$ and $d_i \in N_{D}$, or a facet expression $\psi_d$, for $d \in N_{D}$, or their negation.


Let $\Ra$, $\Rd$, $\mathbf{C}$, $\mathbf{I}$ be denumerable pairwise disjoint sets of abstract role names, concrete role names, concept names, and individual names, respectively.
The set of abstract roles is defined as $\Ra \cup \{ R^- \mid R \in \Ra \} \cup U$, where $U$ is the universal role and $R^-$ is the inverse role of $R$.


A role inclusion axiom (RIA) is an expression of the form $w \sqsubseteq R$, where $w$ is a finite string of roles not including the universal role $U$ and $R$ is an abstract role name distinct from the universal role $U$.

 An abstract role hierarchy $\mathsf{R}_{a}^{H}$ is a finite collection of RIAs. \vipcomment{ A strict partial order $\prec$ on  $\Ra \cup \{ R^- \mid R \in \Ra \}$ is called \emph{a regular order} if $\prec$ satisfies, additionally, $S \prec R$ iff $S^- \prec R$, for all roles R and S.\footnote{We recall that a strict partial order $\prec$  on a set $A$ is an irreflexive and transitive relation on $A$.}


 A RIA $w \sqsubseteq R$ is $\prec$-\emph{regular} if $R$ is an abstract role name, and one of the following conditions holds:

\begin{enumerate}
\item{$ w = RR$,}
\item{$w =R^-$,}
\item{$w=S_1...S_n$ and $S_i \prec R$, for all $1 \leq i \leq n$,}
\item{$w=RS_1...S_n$ and $S_i \prec R$, for all $1 \leq i \leq n$,}
\item{$w=S_1...S_nR$ and $S_i \prec R$, for all $1 \leq i \leq n$.}
\end{enumerate}
An abstract role hierarchy $\mathsf{R}_{a}^{H}$ is regular if there exists a regular order $\prec$ such that each RIA in $\mathsf{R}_{a}^{H}$ is $\prec$-regular.
}

A concrete role hierarchy $\mathsf{R}_{\D}^{H}$ is a finite collection of concrete role inclusion axioms $T_i \sqsubseteq T_j$, where $T_i, T_j \in \Rd$.

A role assertion is an expression of one of the types: $\refl(R)$, $\irref(R)$, $\sym(R)$, $\asym(R)$, $\tra(R)$, and $\mathsf{Dis}(R,S)$, where $R,S \in \Ra \cup \{ R^- \mid R \in \Ra \}$.


Given an abstract role hierarchy $\mathsf{R}_{a}^{H}$ and a set of role assertions $\mathsf{R}^{A}$ without transitivity or symmetry assertions ($\sym(R)$ can be represented by a RIA of type $R^- \sqsubseteq R$ and $\tra(R)$ by $RR\sqsubseteq R$), the set of roles that are \emph{simple} in $\mathsf{R}_{a}^{H} \cup \mathsf{R}^{A}$ is inductively defined as follows: (a)
a role name is simple if it does not occur on the right hand side of a RIA in $\mathsf{R}_{a}^{H}$,
(b) an inverse role $R^-$ is simple if $R$ is, and
(c) if $R$ occurs on the right hand of a RIA in $\mathsf{R}_{a}^{H}$, then $R$ is simple if, for each $w \sqsubseteq R \in \mathsf{R}_{a}^{H}$,  $w=S$, for a simple role $S$.

A set of role assertions $\mathsf{R}^{A}$ is called simple if all roles $R$, $S$ appearing in role assertions of the form $\irref(R)$, $\asym(R)$, or $\mathsf{Dis}(R,S)$ are simple in $\mathsf{R}^{A}$.

An \sroiqd-$RBox$ is a set $\mathsf{R} = \mathsf{R}_{a}^{H} \cup \mathsf{R}_{\D}^{H} \cup \mathsf{R}^{A}$ such that $\mathsf{R}_{a}^{H}$ is a regular abstract role hierarchy, $\mathsf{R}_{\D}^{H}$ is a concrete role hierarchy, and $\mathsf{R}^{A}$ is a finite simple set of role assertions. A formal definition of regular abstract role hierarchy can be found in \cite{Horrocks2006}.

Before introducing the formal definitions of $TBox$ and of $ABox$, we define the set of \sroiqd-concepts as the smallest set such that:
\begin{itemize}
\item[-]{every concept name and the constants $\top$, $\bot$ are concepts,}
\item[-]{if $C$, $D$ are concepts, $R$ is an abstract role (possibly inverse), $S$ is a simple role (possibly inverse), $T$ is a concrete role, $dr$ is a data range for $\D$, $a$ is an individual, and $n$ is a non-negative integer, then $C \sqcap D$, $C \sqcup D$, $\neg C$, $\{a\}$, $\forall R.C$, $\exists R.C$, $\exists S.\mathit{Self}$, $\forall T.dr$, $\exists T.dr$, $\geq nS.C$, and $\leq n S.C$ are also concepts.}
\end{itemize}

A general concept inclusion axiom (GCI) is an expression $C \sqsubseteq D$, where $C$, $D$ are \sroiqd-concepts. An \sroiqd-$TBox$ $\mathcal{T}$ is a finite set of CGIs.

Any expression of one of the following forms: $a : C$, $(a,b) : R$, $(a, e_d) : T$, $(a,b) : \neg R$, $(a,e_d) : \neg T$, $a=b$, $ a \neq b$, where $a,b$ are individuals, $e_d$ is a constant in $N_{C}(d)$, $R$ is a (possibly) inverse abstract role, $P$ is a concrete role, and $C$ is a concept, is called an \emph{individual assertion}. An \sroiqd-$ABox$ $\mathcal{A}$ is a finite set of individual assertions.

An \sroiqd-knowledge base is a triple $\mathcal{K} = (\mathcal{R}, \mathcal{T}, \mathcal{A})$ such that $\mathcal{R}$ is an \sroiqd-$RBox$, $\mathcal{T}$ an \sroiqd-$TBox$, and $\mathcal{A}$ an \sroiqd-$ABox$.
The semantics of \sroiqd\space is given by means of an interpretation $\I= (\Delta^\I, \Delta_{\D}, \cdot^\I)$, where $\Delta^\I$ and $\Delta_{\D}$ are non-empty disjoint domains such that $d^\D\subseteq \Delta_{\D}$, for every $d \in N_{D}$, and
$\cdot^\I$ is an interpretation function.
The interpretation of concepts and roles, axioms and assertions is defined in Table \ref{semsroiq}.

{\small
\begin{longtable}{|>{\centering}m{2.5cm}|c|>{\centering\arraybackslash}m{7.5cm}|}
\hline
Name & Syntax & Semantics \\
\hline

concept & $A$ & $ A^\I \subseteq \Delta^\I$\\

ab. (resp., cn.) rl. & $R$ (resp., $T$ )& $R^\I \subseteq \Delta^\I \times \Delta^\I$ \hspace*{0.5cm} (resp., $T^\I \subseteq \Delta^\I \times \Delta_\D$)\\


\scriptsize{ind. (resp., d. cs.)}& $a$ (resp., $e_{d})$)& $a^\I \in \Delta^\I$ (resp., $e_{d}^{\D} \in d^\D$)\\

nominal & $\{a\}$ & $\{a\}^\I = \{a^\I \}$\\

dtype  (resp., ng.) & $d$ (resp., $\neg d$)& $ d^\D \subseteq \Delta_\D$ (resp., $\Delta_\D \setminus d^\D $)\\



\hline
data range  & $\{ e_{d_1}, \ldots , e_{d_n} \}$& $\{ e_{d_1}, \ldots , e_{d_n} \}^\D = \{e_{d_1}^\D \} \cup \ldots \cup \{e_{d_n}^\D \} $ \\

data range   &  $\psi_d$ & $\psi_d^\D$\\

data range    & $\neg dr$ &  $\Delta_\D \setminus dr^\D $\\

\hline

top (resp., bot.) & $\top$ (resp., $\bot$ )& $\Delta^\I$  (resp., $\emptyset$)\\


negation & $\neg C$ & $(\neg C)^\I = \Delta^\I \setminus C$ \\

conj. (resp., disj.) & $C \sqcap D$ (resp., $C \sqcup D$)& $ (C \sqcap D)^\I = C^\I \cap D^\I$  (resp., $ (C \sqcup D)^\I = C^\I \cup D^\I$)\\


univ.restriction & $\forall R.C$& $(\forall R.C)^\I = \{ x \in \Delta^\I : \forall y \in \Delta^\I \textbf{.} \langle x,y \rangle \in R^\I \rightarrow y \in C^\I \}$ \\

exist. restriction & $\exists R.C$ & $(\exists R.C)^\I = \{x \in \Delta^\I : \exists y \in C^\I \textbf{.} \langle x,y \rangle \in R^\I\}$  \\

self concept & $\exists R.\mathit{Self}$ & $(\exists R.\mathit{Self})^\I = \{ x \in \Delta^\I : \langle x,x \rangle \in R^\I  \}$ \\

datatype exists & $\exists T.dr$ & $(\exists T.dr)^\I = \{x \in \Delta^\I : \exists y \in dr^\D \textbf{.} \langle x,y \rangle \in T^\I\}$\\

datatype value & $\forall T.dr$ & $(\forall T.dr)^\I\hspace*{-0.1cm} = \hspace*{-0.1cm}\{x \in \Delta^\I : \forall y \in \Delta_\D \textbf{.} \langle x,y \rangle \in T^\I\hspace*{-0.1cm} \rightarrow\hspace*{-0.1cm} y \in dr^\D \}$\\

qualified number  &  $\leq_n\!\! R.C$ & $(\leq_n\!\! R.C)^\I = \{ x \in \Delta^\I : |\{ y \in C^\I : \langle x,y \rangle \in R^\I \}| \leq n  \}$\\

restriction& $\geq_n\!\! R.C$ & $(\geq_n\!\! R.C)^\I =\{ x \in \Delta^\I : |\{ y \in C^\I: \langle x,y \rangle \in R^\I  \}| \geq n  \}$ \\

qual. datatype   &  $\leq_n\!\! T.dr$ & $(\leq_n\!\! T.dr)^\I = \{ x \in \Delta^\I : |\{ y \in dr^\D : \langle x,y \rangle \in T^\I \}| \leq n  \}$\\

number restr.& $\geq_n\!\! T.dr$ & $(\geq_n\!\! T.dr)^\I =\{ x \in \Delta^\I : |\{ y \in dr^\D: \langle x,y \rangle \in T^\I  \}| \geq n  \}$ \\

nominals & $\{ a_1, \ldots , a_n \}$& $\{ a_1, \ldots , a_n \}^\I = \{a_1^\I \} \cup \ldots \cup \{a_n^\I \} $ \\

\hline

universal role & U & $(U)^\I = \Delta^\I \times \Delta^\I$\\

inverse role & $R^-$ & $(R^-)^\I = \{\langle y,x \rangle  \mid \langle x,y \rangle \in R^\I\}$\\

\hline

concept subsum. & $C_1 \sqsubseteq C_2$ & $\I \models_\D C_1 \sqsubseteq C_2 \; \Longleftrightarrow \; C_1^\I \subseteq C_1^\I$ \\

ab. role subsum. & $ R_1 \sqsubseteq R_2$ & $\I \models_\D R_1 \sqsubseteq R_2 \; \Longleftrightarrow \; R_1^\I \subseteq R_1^\I$\\

role incl. axiom & $S_1 \ldots S_n \sqsubseteq R$ & $\I \models_\D S_1 \ldots S_n \sqsubseteq R  \; \Longleftrightarrow \; S_1^\I\circ \ldots \circ S_n^\I \subseteq R^\I$\\
cn. role subsum. & $ T_1 \sqsubseteq T_2$ & $\I \models_\D T_1 \sqsubseteq T_2 \; \Longleftrightarrow \; T_1^\I \subseteq T_1^\I$\\

\hline

symmetric role & $\sym(R)$ & $\I \models_\D \sym(R) \; \Longleftrightarrow \; (R^-)^\I \subseteq R^\I$\\

asymmetric role & $\asym(R)$ & $\I \models_\D \asym(R) \; \Longleftrightarrow \; R^\I \cap (R^-)^\I = \emptyset $\\

transitive role & $\tra(R)$ & $\I \models_\D \tra(R) \; \Longleftrightarrow \; R^\I \circ R^\I \subseteq R^\I$\\

disjoint role & $\mathsf{Dis}(R,S)$ & $\I \models_\D \mathsf{Dis}(R,S) \; \Longleftrightarrow \; R^\I \cap S^\I = \emptyset$\\

reflexive role & $\refl(R)$& $\I \models_\D \refl(R) \; \Longleftrightarrow \; \{ \langle x,x \rangle \mid x \in \Delta^\I\} \subseteq R^\I$\\

irreflexive role & $\irref(R)$& $\I \models_\D \irref(R) \; \Longleftrightarrow \; R^\I \cap \{ \langle x,x \rangle \mid x \in \Delta^\I\} = \emptyset  $\\

func. ab. role & $\fun(R)$ & $\I \models_\D \fun(R) \; \Longleftrightarrow \; (R^{-})^\I \circ R^\I \subseteq  \{ \langle x,x \rangle \mid x \in \Delta^\I\}$  \\

func. cn. role & $\fun(T)$ & $\I \models_\D \fun(T) \; \Longleftrightarrow \; \langle x,y \rangle \in T^\I \mbox{ and } \langle x,z \rangle \in T^\I \mbox{ imply } y = z$  \\

\hline

concept assertion & $a : C_1$ & $\I \models_\D a : C_1 \; \Longleftrightarrow \; (a^\I \in C_1^\I) $ \\

agreement & $a=b$ & $\I \models_\D a=b \; \Longleftrightarrow \; a^\I=b^\I$\\

disagreement & $a \neq b$ & $\I \models_\D a \neq b  \; \Longleftrightarrow \; \neg (a^\I = b^\I)$\\


ab. role asser. & $ (a,b) : R $ & $\I \models_\D (a,b) : R \; \Longleftrightarrow \;  \langle a^\I , b^\I \rangle \in R^\I$ \\

cn. role asser. & $ (a,e_d) : T $ & $\I \models_\D (a,e_d) : T \; \Longleftrightarrow \;   \langle a^\I , e_d^\D \rangle \in T^\I$ \\

ng. ab. role asser. & $ (a,b) : \neg R $ & $\I \models_\D (a,b) : \neg R \; \Longleftrightarrow \;   \neg (\langle a^\I , b^\I \rangle \in R^\I)$ \\

ng. cn. role asser. & $ (a,e_d) : \neg T $ &  $\I \models_\D (a,e_d) : \neg T \; \Longleftrightarrow \;  \neg (\langle a^\I , e_d^\D \rangle \in T^\I)$ \\
%
%
%
%
%
%
%
%
%





\hline \caption{Semantics of \sroiqd.}\\
\caption*{\emph{Legenda.} ab: abstract, cn.: concrete, rl.: role, ind.: individual, d. cs.: datatype constant, dtype: datatype, ng.: negated, bot.: bottom, incl.: inclusion, asser.: assertion.}  \label{semsroiq}
\end{longtable}}

Let $\mathcal{A}$, $\mathcal{R}$, $\mathcal{T}$ be, respectively, an \sroiqd-$ABox$, an \sroiqd-$RBox$, and  an \sroiqd-$TBox$. An interpretation $\I= (\Delta ^ \I, \Delta_{\D}, \cdot ^ \I)$ is a $\D$-model of $\mathcal{R}$ (resp., $\mathcal{T}$), and we write $\I \models_{\D} \mathcal{R}$ (resp., $\I \models_{\D} \mathcal{T}$), if $\I$ satisfies each axiom in $\mathcal{R}$ (resp., $\mathcal{T}$) according to the semantic rules in Table \ref{semsroiq}.  Analogously,  $\I= (\Delta^ \I, \Delta_{\D}, \cdot^\I)$ is a $\D$-model of $\mathcal{A}$, and we write $\I \models_{\D} \mathcal{A}$, if $\I$ satisfies each assertion in $\mathcal{A}$, according to the semantic rules in Table \ref{semsroiq}.

An \sroiqd-knowledge base $\mathcal{K}=(\mathcal{A}, \mathcal{T}, \mathcal{R})$ is consistent if there is an interpretation $\I= (\Delta^ \I, \Delta_{\D}, \cdot^\I)$ that is a $\D$-model of $\mathcal{A}$,  $\mathcal{T}$, and $\mathcal{R}$.

Decidability of the consistency problem for \sroiqd-knowledge bases was proved in \cite{Horrocks2006} by means of a tableau-based decision procedure and its computational complexity was shown to be \textbf{N2EXPTime}-complete in \cite{Kazakov:08:RIQ:SROIQ}.

\section{The logic $\dlss$}\label{dlss}
In this section we introduce the description logic $\dlss$ (shortly referred to as $\shdlss$) and prove that the consistency problem for $\shdlss$-knowledge bases is decidable by reducing it to the satisfiability problem for $\flqsr$-formulae. Then we show that under certain restrictions the consistency problem for $\shdlss$-knowledge bases is \textbf{NP}-complete. Finally we briefly illustrate how SWRL-rules can be translated into the language of $\flqsr$.

Let $\D$, $\Ra$, $\Rd$, $\I$, $\C$ be as in Section \ref{sroiqd}.

\noindent
(a) $\shdlss$-datatype, (b) $\shdlss$-concept, (c) $\shdlss$-abstract role, and (d) $\shdlss$-concrete role terms are constructed according to the following syntax rules:
\begin{itemize}
\item[(a)] $t_1, t_2 \longrightarrow dr ~|~\neg t_1 ~|~t_1 \sqcap t_2 ~|~t_1 \sqcup t_2 ~|~\{e_{d}\}\, ,$

\item[(b)] $C_1, C_ 2 \longrightarrow A ~|~\top ~|~\bot ~|~\neg C_1 ~|~C_1 \sqcup C_2 ~|~C_1 \sqcap C_2 ~|~\{a\} ~|~\exists R.\mathit{Self}| \exists R.\{a\}| \exists P.\{e_{d}\}\, ,$

\item[(c)] $R_1, R_2 \longrightarrow S ~|~U ~|~R_1^{-} ~|~ \neg R_1 ~|~R_1 \sqcup R_2 ~|~R_1 \sqcap R_2 ~|~R_{C_1 |} ~|~R_{|C_1} ~|~R_{C_1 ~|~C_2} ~|~id(C)\, ,$

\item[(d)] $P \longrightarrow T ~|~\neg P ~|~P_{C_1 |} ~|~P_{|t_1} ~|~P_{C_1 | t_1}\, ,$
\end{itemize}
where $dr$ is a data range for $\D$, $t_1,t_2$ are datatype terms, $e_{d}$ is a constant in $N_{C}(d)$, $a$ is an individual name, $A$ is a concept name, $C_1, C_2$ are $\shdlss$-concept terms, $S$ is an abstract role name, $R, R_1,R_2$ are $\shdlss$-abstract role terms, $T$ a concrete role name, and $P$ a $\shdlss$-concrete role term.

A $\shdlss$-knowledge base is a triple ${\mathcal K} = (\mathcal{R}, \mathcal{T}, \mathcal{A})$ such that $\mathcal{R}$ is a $\shdlss$-$RBox$, $\mathcal{T}$ is a $\shdlss$-$TBox$, and $\mathcal{A}$ a $\shdlss$-$ABox$. A $\shdlss$-$RBox$ is a collection of statements of the following forms: $R_1 \equiv R_2$, $R_1 \sqsubseteq R_2$, $R_1\ldots R_n \sqsubseteq R_{n+1}$, $\sym(R_1)$, $\asym(R_1)$, $\refl(R_1)$, $\irref(R_1)$, $\mathsf{Dis}(R_1,R_2)$,
$\tra(R_1)$, $\fun(R_1)$, $P_1 \equiv P_2$, $P_1 \sqsubseteq P_2$, $\fun(P_1)$, where $R_1,R_2$ are $\shdlss$-abstract role terms and $P_1,P_2$ are $\shdlss$-concrete role terms. A $\shdlss$-$TBox$ is a set of statements of the types:
\begin{itemize}
\item[-] $C_1 \equiv C_2$, $C_1 \sqsubseteq C_2$, $C_1 \sqsubseteq \forall R_1.C_2$, $\exists R_1.C_1 \sqsubseteq C_2$, $\geq_n\!\! R_1. C_1 \sqsubseteq C_2$, \\$C_1 \sqsubseteq {\leq_n\!\! R_1. C_2}$,
\item[-] $t_1 \equiv t_2$, $t_1 \sqsubseteq t_2$, $C_1 \sqsubseteq \forall P_1.t_1$, $\exists P_1.t_1 \sqsubseteq C_1$, $\geq_n\!\! P_1. t_1 \sqsubseteq C_1$, $C_1 \sqsubseteq {\leq_n\!\! P_1. t_1}$,
\end{itemize}
where $C_1,C_2$ are $\shdlss$-concept terms, $t_1,t_2$ datatype terms, $R_1$  a $\shdlss$-abstract role term, $P_1$ a $\shdlss$-concrete role term.

A $\shdlss$-$ABox$ is a set of assertions of the forms: $a : C_1$, $(a,b) : R_1$, $(a,b) : \neg R_1$, $a=b$, $a \neq b$, $e_{d} : t_1$, $(a, e_{d}) : P_1$, $(a, e_{d}) : \neg P_1$, 
where $C_1$ is a $\shdlss$-concept term, $d$ is a datatype, $t_1$ is a datatype term, $R_1$ is a $\shdlss$-abstract role term, $P_1$ is a $\shdlss$-concrete role term, $a,b$ are individual names, and $e_{d}$ is a constant in $N_{C}(d)$.

The semantics of $\shdlss$ is similar to that of \sroiqd\ (cf.\  Section \ref{sroiqd}). The interpretation of terms, axioms, and assertions of $\shdlss$ shared with \sroiqd\space is illustrated in Table \ref{semsroiq} while the semantics of terms and statements specific to $\shdlss$ is described in Table \ref{semdlss}. The notions of $\D$-model of a $\shdlss$-$RBox$, $\shdlss$-$TBox$, $\shdlss$-$ABox$, and the notion of consistency of a $\shdlss$-knowledge base are similar to the ones described in Section \ref{sroiqd} for \sroiqd.
{\small
\begin{longtable}{|>{\centering}m{4cm}|c|c|}
\hline
Name & Syntax & Semantics \\
\hline

data range $dr$ & $ dr $ & $  dr^{\D}\subseteq \Delta_{\D} $ \\

negative datatype term & $ \neg t_1 $ & $  (\neg t_1)^{\D} = \Delta_{\D} \setminus t_1^{\D}$ \\

datatype terms intersection & $ t_1 \sqcap t_2 $ & $  (t_1 \sqcap t_2)^{\D} = t_1^{\D} \cap t_2^{\D} $ \\

datatype terms union & $ t_1 \sqcup t_2 $ & $  (t_1 \sqcup t_2)^{\D} = t_1^{\D} \cup t_2^{\D} $ \\

constant in $N_{C}(d)$ & $ e_{d} $ & $ e_{d}^\D \in d^\D$ \\

\hline

valued exist. quantification & $\exists R.{a}$ & $(\exists R.{a})^\I = \{ x \in \Delta^\I : \langle x,a^\I \rangle \in R^\I  \}$ \\

datatyped exist. quantif. & $\exists P.{e_{d}}$ & $(\exists P.e_{d})^\I = \{ x \in \Delta^\I : \langle x, e^\D_{d} \rangle \in P^\I  \}$ \\

\hline

abstract role complement & $ \neg R $ & $ (\neg R)^\I=(\Delta^\I \times \Delta^\I) \setminus R^\I $\\

abstract role union & $R_1 \sqcup R_2$ & $ (R_1 \sqcup R_2)^\I = R_1^\I \cup R_2^\I $\\

abstract role intersection & $R_1 \sqcap R_2$ & $ (R_1 \sqcap R_2)^\I = R_1^\I \cap R_2^\I $\\

abstract role domain restr. & $R_{C \mid }$ & $ (R_{C \mid })^\I = \{ \langle x,y \rangle \in R^\I : x \in C^\I  \} $\\

concrete role complement & $ \neg P $ & $ (\neg P)^\I=(\Delta^\I \times \Delta^\D) \setminus P^\I $\\

concrete role domain restr. & $P_{C \mid }$ & $ (P_{C \mid })^\I = \{ \langle x,y \rangle \in P^\I : x \in C^\I  \} $\\

concrete role range restr. & $P_{ \mid t}$ &  $ (P_{\mid t})^\I = \{ \langle x,y \rangle \in P^\I : y \in t^\D  \} $\\

concrete role restriction & $P_{ C_1 \mid t}$ &  $ (P_{C_1 \mid t})^\I = \{ \langle x,y \rangle \in P^\I : x \in C_1^\I \wedge y \in t^\D  \} $\\

\hline

datatype terms equivalence & $ t_1 \equiv t_2 $ & $ \I \models_{\D} t_1 \equiv t_2 \Longleftrightarrow t_1^{\D} = t_2^{\D}$\\

datatype terms diseq. & $ t_1 \not\equiv t_2 $ & $ \I \models_{\D} t_1 \not\equiv t_2 \Longleftrightarrow t_1^{\D} \neq t_2^{\D}$\\

datatype terms subsum. & $ t_1 \sqsubseteq t_2 $ &  $ \I \models_{\D} (t_1 \sqsubseteq t_2) \Longleftrightarrow t_1^{\D} \subseteq t_2^{\D} $ \\

%
%
%
\hline \caption{Semantics of terms and axioms specific to $\shdlss$.} \label{semdlss}
\end{longtable}
}


In the following theorem we prove the decidability of the consistency problem for $\shdlss$-knowledge bases.

\begin{theorem}\label{theorem}
Let $\mathcal{K}$ be a $\shdlss$-knowledge base. Then, one can construct a \flqsr-formula $\varphi_{\mathcal{K}}$ s.t. $\varphi_{\mathcal{K}}$ is satisfiable if and only if $\mathcal{K}$ is consistent.  \end{theorem}
\begin{proof}
As a preliminary step, observe that the statements of the $\shdlss$-knowledge base $\mathcal{K}$ that need to be considered  are those of the following types:
\begin{itemize}
\item[-] $C_1 \equiv \top$, $C_1 \equiv \neg C_2$, $C_1 \equiv C_2 \sqcup C_3$, $C_1 \equiv \{a\}$, $C_1 \sqsubseteq \forall R_1.C_2$, $\exists R_1.C_1 \sqsubseteq C_2$, $\geq_n\!\! R_1. C_1 \sqsubseteq C_2$, $C_1 \sqsubseteq {\leq_n\!\! R_1. C_2}$, $C_1 \sqsubseteq \forall P_1.t_1$, $\exists P_1.t_1 \sqsubseteq C_1$, $\geq_n\!\! P_1. t_1 \sqsubseteq C_1$, $C_1 \sqsubseteq {\leq_n\!\! P_1. t_1}$,
\item[-] $R_1 \equiv U$, $R_1 \equiv \neg R_2$, $R_1 \equiv R_2 \sqcup R_3$, $R_1 \equiv R_2^{-}$, $R_1 \equiv id(C_1)$, $R_1 \equiv R_{2_{C_1 |}}$, $R_1 \ldots R_n \sqsubseteq R_{n+1}$, $\refl(R_1)$, $\irref(R_1)$, $\mathsf{Dis}(R_1,R_2)$, $\fun(R_1)$,
\item[-] $P_1 \equiv P_2$, $P_1 \equiv \neg P_2$, $P_1 \sqsubseteq P_2$, $\fun(P_1)$, $P_1 \equiv P_{2_{C_1 |}}$, $P_1 \equiv P_{2_{C_1 | t_1}}$, $P_1 \equiv P_{2_{| t_1}}$,
\item[-] $t_1 \equiv t_2$, $t_1 \equiv \neg t_2$, $t_1 \equiv t_2 \sqcup t_3$, $t_1 \equiv \{e_d\}$,
\item[-] $a : C_1$, $(a,b) : R_1$, $(a,b) : \neg R_1$, $a=b$, $a \neq b$,
$e_{d} : t_1$, $(a, e_{d}) : P_1$, $(a, e_{d}) : \neg P_1$.
\end{itemize}
In order to define the \flqsr-formula $\varphi_{\mathcal{K}}$, we shall make use of a mapping $\tau$ from the $\shdlss$-statements (and their conjunctions) listed above into \flqsr-formulae. To prepare for the definition of $\tau$, we map injectively individuals $a$ and constants $e_d \in N_{C}(d)$ into level $0$ variables $x_a$ and $x_{e_d}$, the constant concepts $\top$ and $\bot$, datatype terms $t$, and concept terms $C$ into level $1$ variables $X_{\top}^1$, $X_{\bot}^1$, $X_{t}^1$, $X_{C}^1$, respectively, and the universal relation on individuals $U$, abstract role terms $R$, and concrete role terms $P$ into level $3$ variables $X_{U}^3$, $X_{R}^3$, and $X_{P}^3$, respectively.\footnote{The use of level $3$ variables to model abstract and concrete role terms is motivated by the fact that their elements, that is ordered pairs $\langle x, y \rangle$ are encoded in Kuratowski's style as $\{\{x\}, \{x,y\}\}$, namely as collections of sets of objects. Variables of level $2$ are used in the formulae $\psi_8$ and $\psi_9$ of the construction to model the fact that level $3$ variables representing role terms are binary relations.}

Then the mapping $\tau$ is defined as follows:
\smallskip

%
%
%
{\small
\noindent $\tau(C_1 \equiv \top) \defAs (\forall z)(z \in X_{C_1}^1 \leftrightarrow z \in X_{\top}^1)$,

\noindent $\tau(C_1 \equiv \neg C_2) \defAs (\forall z)(z \in X_{C_1} \leftrightarrow \neg(z \in X_{C_2}^1))$,

\noindent $\tau(C_1 \equiv C_2 \sqcup C_3 ) \defAs (\forall z)(z \in X_{C_1}^1 \leftrightarrow (z \in X_{C_2}^1 \vee z \in X_{C_3}^1))$,

\noindent $\tau(C_1 \equiv \{a\}) \defAs (\forall z)(z \in X_{C_1}^1 \leftrightarrow z = x_a)$,

\noindent $\tau(C_1 \sqsubseteq \forall R_1.C_2) \defAs (\forall z_1)(\forall z_2)(z_1 \in X_{C_1}^1 \rightarrow (\langle z_1,z_2 \rangle \in X_{R_1}^3 \rightarrow z_2 \in X_{C_2}^1))$,

\noindent $\tau(\exists R_1.C_1 \sqsubseteq C_2) \defAs (\forall z_1)(\forall z_2)((\langle z_1,z_2 \rangle \in X_{R_1}^3 \wedge z_2 \in X_{C_1}^1) \rightarrow z_1 \in X_{C_2}^1)$,

\noindent $\tau(C_1 \equiv \exists R_1.\{a\}) \defAs(\forall z)(z \in X_{C_1}^1 \leftrightarrow \langle z,x_{a}\rangle \in X_{R_1}^3)$,

\noindent $\tau(C_1 \sqsubseteq \leq_n\!\! R_1.C_2) \defAs (\forall z)(\forall z_1)\ldots (\forall z_{n+1})(z \in X_{C_1}^1 \rightarrow$

\noindent $\hfill  ( \overset{n+1}{\underset{i=1}\bigwedge}(z_i \in X_{C_2} \wedge \langle z,z_i\rangle \in X_{R_1}^3) \rightarrow \underset {i<j} {\bigvee} z_i = z_j))$,

\noindent $\tau(\geq_n\!\! R_1.C_1 \sqsubseteq C_2) \defAs (\forall z)(\forall z_1)\ldots (\forall z_{n})( \overset {n}{ \underset{i=1}\bigwedge }((z_i \in X_{C_1}^1 \wedge \langle z,z_i\rangle \in X_{R_1}^3) \rightarrow$

\noindent $\hfill  \underset {i<j} \bigwedge z_i \neq z_j) \rightarrow z \in X_{C_2}^1)$,

\noindent $\tau(C_1 \sqsubseteq \forall P_1.t_1) \defAs (\forall z_1)(\forall z_2)(z_1 \in X_{C_1}^1 \rightarrow (\langle z_1,z_2 \rangle \in X_{P_1}^3 \rightarrow z_2 \in X_{t_1}^1))$,

\noindent $\tau(\exists P_1.t_1 \sqsubseteq C_1) \defAs (\forall z_1)(\forall z_2)((\langle z_1,z_2 \rangle \in X_{P_1}^3 \wedge z_2 \in X_{t_1}^1) \rightarrow z_1 \in X_{C_1}^1)$,

\noindent $\tau(C_1 \equiv \exists P_1.\{e_{d}\}) \defAs (\forall z)(z \in X_{C_1}^1 \leftrightarrow \langle z,x_{e_{d}}\rangle \in X_{P_1}^3)$,

\noindent $\tau(C_1 \sqsubseteq \leq_n\!\! P_1.t_1) \defAs (\forall z)(\forall z_1)\ldots (\forall z_{n+1})(z \in X_{C_1}^1 \rightarrow$

\noindent $\hfill ( \overset {n+1} { \underset{i=1}\bigwedge }(z_i \in X_{t_1} \wedge \langle z,z_i\rangle \in X_{P_1}^3) \rightarrow \underset{i<j} {\bigvee} z_i = z_j))$,

\noindent $\tau(\geq_n\!\! P_1.t_1 \sqsubseteq C_1) \defAs$

\noindent $\hfill (\forall z)(\forall z_1)\ldots (\forall z_{n})( \overset {n} { \underset {i=1}\bigwedge}((z_i \in X_{t_1}^1 \wedge \langle z,z_i\rangle \in X_{P_1}^3) \rightarrow \underset {i<j} {\bigwedge} z_i \neq z_j) \rightarrow z \in X_{C_1}^1)$,

\noindent $\tau(R_1 \equiv U) \defAs (\forall Z^2)(Z^2 \in X_{R_1}^3 \leftrightarrow Z^2 \in X_{U}^3)$,

\noindent $\tau(R_1 \equiv \neg R_2) \defAs (\forall z_1)(\forall z_2)(\langle z_1,z_2\rangle \in X_{R_1}^3 \leftrightarrow \neg (\langle z_1,z_2\rangle \in X_{R_2}^3 ))$,

\noindent $\tau(R_1 \equiv R_2 \sqcup R_3) \defAs (\forall Z^2)(Z^2 \in X_{R_1}^3 \leftrightarrow (Z^2 \in X_{R_2}^3 \vee Z^2 \in X_{R_3}^3))$,

\noindent $\tau(R_1 \equiv R_2^{-}) \defAs (\forall z_1)(\forall z_2)(\langle z_1,z_2\rangle \in X_{R_1}^3 \leftrightarrow \langle z_2,z_1\rangle \in X_{R_2}^3 ))$,

\noindent $\tau(R_1 \equiv id(C_1)) \defAs (\forall z_1)(\forall z_2)(\langle z_1,z_2\rangle \in X_{R_1}^3 \leftrightarrow (z_1 \in X_{C_1}^1 \wedge z_2 \in X_{C_1}^1 \wedge z_1 =z_2))$,


\noindent $\tau(R_1 \equiv R_{2_{C_1 |}}) \defAs (\forall z_1)(\forall z_2)(\langle z_1,z_2\rangle \in X_{R_1}^3 \leftrightarrow  (\langle z_1,z_2\rangle \in X_{R_2}^3 \wedge z_1 \in X_{C_1}^1))$,

\noindent $\tau(R_1 \ldots R_n \sqsubseteq R_{n+1}) \defAs (\forall z)(\forall z_1)\ldots (\forall z_{n})$

\noindent $\hfill ((\langle z, z_1\rangle \in X_{R_1}^3 \wedge \ldots \wedge \langle z_{n-1},z_n\rangle \in X_{R_{n}}^3) \rightarrow  \langle z, z_n\rangle\in X_{R_{n+1}}^3)$,

\noindent $\tau(\refl(R_1)) \defAs (\forall z)(\langle z, z\rangle \in X_{R_1}^3)$,

\noindent $\tau(\irref(R_1)) \defAs (\forall z)(\neg (\langle z,z\rangle \in X_{R_1}^3))$,

\noindent $\tau(\fun(R_1)) \defAs (\forall z_1)(\forall z_2)(\forall z_3)((\langle z_1,z_2\rangle \in X_{R_1}^3 \wedge \langle z_1,z_3\rangle \in X_{R_1}^3) \rightarrow z_2 =z_3)$,

\noindent $\tau(P_1 \equiv P_2) \defAs (\forall Z^2)(Z^2 \in X_{P_1}^3 \leftrightarrow Z^2 \in X_{P_2}^3)$,

\noindent $\tau(P_1 \equiv \neg P_2) \defAs (\forall Z^2)(Z^2 \in X_{P_1}^3 \leftrightarrow \neg(Z^2 \in X_{P_2}^3))$,

\noindent $\tau(P_1 \sqsubseteq P_2) \defAs (\forall Z^2)(Z^2 \in X_{P_1}^3 \rightarrow Z^2 \in X_{P_2}^3)$,

\noindent $\tau(\fun(P_1)) \defAs (\forall z_1)(\forall z_2)(\forall z_3)((\langle z_1,z_2\rangle \in X_{P_1}^3 \wedge \langle z_1,z_3\rangle \in X_{P_1}^3) \rightarrow z_2 =z_3)$,

\noindent $\tau(P_1 \equiv P_{2_{C_1 |}}) \defAs (\forall z_1)(\forall z_2)(\langle z_1,z_2\rangle \in X_{P_1}^3 \leftrightarrow (\langle z_1,z_2\rangle \in X_{P_2}^3 \wedge z_1 \in X_{C_1}^1))$,

\noindent $\tau(P_1 \equiv P_{2_{|t_1}}) \defAs (\forall z_1)(\forall z_2)(\langle z_1,z_2\rangle \in X_{P_1}^3 \leftrightarrow (\langle z_1,z_2\rangle \in X_{P_2}^3 \wedge z_2 \in X_{t_1}^1))$,

\noindent $\tau(P_1 \equiv P_{2_{C_1|t_1}}) \defAs (\forall z_1)(\forall z_2)(\langle z_1,z_2\rangle \in X_{P_1}^3 \leftrightarrow$

\noindent $\hfill (\langle z_1,z_2\rangle \in X_{P_2}^3 \wedge z_1 \in X_{C_1}^1 \wedge z_2 \in X_{t_1}^1))$,

\noindent $\tau(t_1 \equiv t_2)\defAs (\forall z)(z \in X_{t_1}^1 \leftrightarrow z \in X_{t_2}^1)$,

\noindent $\tau(t_1 \equiv \neg t_2)\defAs (\forall z)(z \in X_{t_1}^1 \leftrightarrow \neg (z \in X_{t_2}^1))$,

\noindent $\tau(t_1 \equiv t_2 \sqcup t_3)\defAs (\forall z)(z \in X_{t_1}^1 \leftrightarrow (z \in X_{t_2}^1\vee z \in X_{t_3}^1))$,

\noindent $\tau(t_1 \equiv t_2 \sqcap t_3)\defAs (\forall z)(z \in X_{t_1}^1 \leftrightarrow (z \in X_{t_2}^1\wedge z \in X_{t_3}^1))$,

\noindent $\tau(t_1 \equiv \{e_{d}\})\defAs (\forall z)(z \in X_{t_1}^1 \leftrightarrow z = x_{e_{d}})$,

\noindent $\tau(a : C_1) \defAs x_a \in X_{C_1}^1$,

\noindent $\tau((a,b) : R_1) \defAs \langle x_a, x_b\rangle \in X_{R_1}^3$,

\noindent $\tau((a,b) : \neg R_1) \defAs \neg(\langle x_a, x_b\rangle \in X_{R_1}^3)$,

\noindent $\tau(a=b) \defAs x_a = x_b$,

\noindent $\tau(a\neq b) \defAs \neg (x_a = x_b)$,

\noindent $\tau(e_d : t_1) \defAs x_{e_d} \in X_{t_1}^1$,

\noindent $\tau((a,e_d) : P_1) \defAs \langle x_a, x_{e_d}\rangle \in X_{P_1}^3$,

\noindent $\tau((a,e_d) : \neg P_1) \defAs \neg(\langle x_a, x_{e_d}\rangle \in X_{P_1}^3)$,

\noindent $\tau(\alpha \wedge \beta) \defAs \tau(\alpha) \wedge \tau(\beta)$.
}

\smallskip

Let $\mathcal{K}$ be our $\shdlss$-knowledge base, and let $\ck$, $\ark$, $\crk$, and $\ik$ be, respectively, the sets of concept, of abstract role, of concrete role, and of individual names in $\mathcal{K}$. Moreover, let $N_{D}^\mathcal{K} \subseteq N_{D}$ be the set of datatypes in $\mathcal{K}$, $N_{F}^\mathcal{K}$ a restriction of $N_{F}$ assigning to every $d \in N_{\D}^\mathcal{K}$ the set $N_{F}^\mathcal{K}(d)$ of facets in $N_{F}(d)$ and in $\mathcal{K}$. Analogously, let $N_{C}^{\mathcal{K}}$ be a restriction of the function $N_{C}$ associating to every $d \in N_{\D}^\mathcal{K}$ the set $N_{C}^\mathcal{K}(d)$ of constants  contained in $N_{C}(d)$ and in $\mathcal{K}$. Finally, for every datatype $d \in N_{D}^\mathcal{K}$, let $\bfk(d)$ be the set of facet expressions for $d$ occurring in $\mathcal{K}$ and not in $N_{F}(d) \cup \{\top^{d},\bot_{d}\}$. We define the \flqsr-formula $\varphi_{\mathcal{K}}$ expressing the consistency of $\mathcal{K}$ as follows: {\small
\[
\varphi_{\mathcal{K}} \defAs \bigwedge_{i=1}^{12}\psi_i \wedge \underset {H \in \mathcal{K}}\bigwedge \tau(H) \, ,
\]}
where

\begin{itemize}{\small
\item[-] $\psi_1 \defAs (\forall z)(z \in X_{\I}^1 \leftrightarrow \neg(z \in X_{\D}^1))\wedge (\forall z)(z \in X_{\I}^1 \vee z \in X_{\D}^1)\wedge$

$\hfill \neg (\forall z)\neg (z \in X_{\I}^1) \wedge \neg (\forall z)\neg (z \in X_{\D}^1)$,

\item[-] $\psi_2 \defAs ((\forall z)(z \in X_{\I}^1 \leftrightarrow z \in X_{\top}^1) \wedge (\forall z)\neg (z \in X_{\bot})$,

\item[-] $\psi_3 \defAs \underset{A \in \ck}\bigwedge (\forall z)(z \in X_{A}^1 \rightarrow z \in X_{\I}^1)$,

\item[-] $\psi_4 \defAs ( \underset{d \in N_{D}^\mathcal{K}}\bigwedge((\forall z)(z \in X_{d}^1 \rightarrow z \in X_{\D}^1) \wedge \neg (\forall z)\neg(z \in X_{d}^1))$

$\hfill \wedge (\forall z) (\underset{(d_i,d_j \in N_{D}^\mathcal{K}, i < j)}\bigwedge (z \in X_{d_i}^1 \leftrightarrow \neg (z \in X_{d_j}^1))))$,

\item[-] $\psi_5 \defAs \underset{d \in N_{D}^\mathcal{K}}\bigwedge((\forall z)(z \in X_{d}^1 \leftrightarrow z \in X_{\top_d}^1) \wedge (\forall z)\neg(z \in X_{\bot_d}^1))$,

\item[-] $\psi_6 \defAs \underset{d \in N_{D}^\mathcal{K}} \bigwedge  \quad \underset {f_d \in N_{F}^\mathcal{K}(d)}\bigwedge (\forall z)(z \in X_{f_d}^1 \rightarrow z \in X_{d}^1)$,

\item[-] $\psi_7 \defAs (\forall z_1)(\forall z_2)((z_1 \in X_{\I}^1 \wedge z_2 \in X_{\I}^1) \leftrightarrow \langle z_1,z_2 \rangle \in X_{U}^3)$,

\item[-] $\psi_8 \defAs  \underset {R \in \ark} \bigwedge((\forall Z^2)(Z^2 \in X_R^3 \rightarrow \neg (\forall z_1)(\forall z_2) \neg (\langle z_1,z_2\rangle = Z^2))$

$\hfill \wedge (\forall z_1)(\forall z_2)(\langle z_1,z_2\rangle \in X_R^3 \rightarrow (z_1 \in X_{\I}^1 \wedge z_2 \in X_{\I}^1)))$,

\item[-] $\psi_9 \defAs \underset{T \in \crk} \bigwedge ((\forall Z^2)(Z^2 \in X_T^3 \rightarrow \neg (\forall z_1)(\forall z_2) \neg (\langle z_1,z_2\rangle = Z^2))$

$\hfill \wedge (\forall z_1)(\forall z_2)(\langle z_1,z_2\rangle \in X_T^3 \rightarrow (z_1 \in X_{\I}^1 \wedge z_2 \in X_{\D}^1)))$,

\item[-] $\psi_{10} \defAs \underset{a \in \ik} \bigwedge(x_a \in X_{\I}^1) \wedge \underset {d \in N_{D}^\mathcal{K}} \bigwedge \quad \underset {e_d \in N_{C}^\mathcal{K}(d)} \bigwedge x_{e_d} \in X_{d}^1$,

\item[-] $\psi_{11} \defAs \underset {e_{d_1},\ldots, e_{d_n} \textrm{ in } \mathcal{K}} \bigwedge (\forall z)(z \in X_{\{e_{d_1},\ldots, e_{d_n}\}}^1\leftrightarrow \overset{n}  { \underset {i=1} \bigvee }(z = x_{e_{d_i}}))$

$\hfill \wedge \quad \underset {a_{1},\ldots, a_{n} \textrm{ in } \mathcal{K}} \bigwedge (\forall z)(z \in X_{\{a_{1},\ldots, a_{n}\}}^1 \leftrightarrow  \overset {n}{\underset {i=1} \bigvee}(z = x_{a_{i}}))$,

\item[-] $\psi_{12} \defAs \underset {d \in N_{\D}^\mathcal{K}} \bigwedge  \quad \underset {\psi_d \in \bfk(d)} \bigwedge (\forall z)(z \in X_{\psi_d}^1 \leftrightarrow z \in \sigma(X_{\psi_d}^1))$,
}

with $\sigma$ the transformation function from \flqsr-variables of level 1 to \flqsr-formulae recursively defined, for $d \in N_\D^\mathcal{K}$, by
\[ {\small
\sigma(X_{\psi_d}^1) \defAs \begin{cases}
X_{\psi_d}^1 & \text{if } \psi_d \in N_{F}^\mathcal{K}(d) \cup \{\top^{d},\bot_{d}\}\\
\neg \sigma(X_{\chi_d}^1) & \text{if }  \psi_d = \neg \chi_d\\
\sigma(X_{\chi_d}^1) \wedge \sigma(X_{\varphi_d}^1) & \text{if } \psi_d = \chi_d \wedge \varphi_d\\
\sigma(X_{\chi_d}^1) \vee \sigma(X_{\varphi_d}^1) & \text{if } \psi_d = \chi_d \vee \varphi_d\,.
\end{cases} }
\]
%
%
%

\end{itemize}

\noindent In the above formulae, the variable $X_{\I}^1$ denotes the set of individuals $\I$, $X_{d}^1$ a datatype $d \in N_{D}^\mathcal{K}$, $X_{\D}^1$  a superset of the union of datatypes in $ N_{D}^\mathcal{K}$, $X_{\top_d}^1$ and $X_{\bot_d}^1$ the constants $\top_d$ and $\bot_d$, and $X_{f_d}^1$, $X_{\psi_d}^1$ a facet $f_d$ and a facet expression $\psi_d$, for $d \in N_{D}^\mathcal{K}$, respectively. In addition, $X_{A}^1$, $X_{R}^3$, $X_{T}^3$ denote a concept name $A$, an abstract role name $R$, and a concrete role name $T$ occurring in $\mathcal{K}$, respectively. Finally, $X_{\{e_{d_1},\ldots,e_{d_n}\}}^1$ denotes a data range $\{e_{d_1},\ldots,e_{d_n}\}$ occurring in $\mathcal{K}$, and $X_{\{a_{1},\ldots,a_{n}\}}^1$ a finite set $\{a_1,\ldots,a_n\}$ of nominals in $\mathcal{K}$.

Clearly, the constraints $\psi_1$-$\psi_{12}$ have been introduced to guarantee that each model of $\varphi_{\mathcal{K}}$ can be easily transformed into a $\shdlss$-interpretation.

Next we show that the consistency problem for $\mathcal{K}$ is equivalent to the satisfiability problem for $\varphi_{\mathcal{K}}$.

Let us first assume that $\varphi_{\mathcal{K}}$ is satisfiable. It is not hard to see that $\varphi_{\mathcal{K}}$ is satisfied by a \flqsr-model of the form $\mathbfcal{M} = (D_1 \cup D_2, M)$, where:

\medskip
- $D_1$ and $D_2$ are disjoint nonempty sets and $\underset{d \in N_{D}^\mathcal{K}}\bigcup d^{\D} \subseteq D_2$,

- $MX_{\I}^1 \defAs D_1$, $MX_{\D}^1 \defAs D_2$,

- $MX_{d}^1 \defAs d^{\D}$, for every $d \in N_{D}^\mathcal{K}$,

- $MX_{f_d}^1 \defAs f_d^{\D}$, for every $f_d \in N_{F}^\mathcal{K}(d)$, with $d \in N_{D}^\mathcal{K}$.

\smallskip
\noindent
Exploiting the fact that $\mathbfcal{M}$ satisfies the constraints $\psi_1$-$\psi_{12}$, it is then possible to define a $\shdlss$-interpretation $\I_{\mathbfcal{M}} = (\Delta^{\I}, \Delta_{\D}, \cdot^{\I})$, by putting $\Delta^{\I} \defAs MX_{\I}^1$, $\Delta_{\D} \defAs MX_{\D}^1$, $A^\I \defAs MX_{A}^1$, for every concept name $A \in \ck$, $S^\I \defAs MX_{S}^3$, for every abstract role name $S \in \ark$, $T^\I \defAs MX_{T}^3$, for every concrete role name $T \in \crk$, and $a^\I \defAs Mx_{a}$, for every individual $a \in \ik$.

Since $\mathbfcal{M} \models \underset {H \in \mathcal{K}}\bigwedge \tau(H)$ and, as can be easily checked, $\I_{\mathbfcal{M}} \models_{\D} H$ if and only if $\mathbfcal{M} \models \tau(H)$, for every statement $H \in \mathcal{K}$, we plainly have $\I_{\mathbfcal{M}} \models_{\D} \mathcal{K}$, namely $\mathcal{K}$ is consistent, as we wished to prove.

Conversely, let $\mathcal{K}$ be a consistent $\shdlss$-knowledge base. Then, there is a $\shdlss$-interpretation $\I = (\Delta^{\I}, \Delta_{\D}, \cdot^{\I})$ such that $\I \models_{\D} \mathcal{K}$. We show how to construct, out of the datatype map $\D$ and the $\shdlss$-interpretation $\I$, a \flqsr-interpretation $\mathbfcal{M}_{\I,\D} = (D_{\I,\D}, M_{\I,\D})$ which satisfies $\varphi_{\mathcal{K}}$. Let us put $D_{\I,\D} \defAs \Delta^{\I} \cup \Delta_{\D}$ and define $M_{\I,\D}$ by putting $M_{\I,\D} X_{\I}^1 \defAs \Delta^{\I}$, $M_{\I,\D} X_{\D}^1 \defAs \Delta_{\D}$,  $M_{\I,\D} X_{U}^3 \defAs U^{\I}$, $M_{\I,\D} X_{dr}^1 \defAs dr^{\D}$, for every variable $X_{dr}^1$ in $\varphi$ denoting a data range $dr$ occurring in $\mathcal{K}$, $M_{\I,\D} X_{A}^1 \defAs A^{\I}$, for every $X_{A}^1$ in $\varphi$ denoting a concept name in $\mathcal{K}$, and $M_{\I,\D} X_{S}^3 \defAs S^{\I}$, for every $X_{S}^3$ in $\varphi$ denoting an abstract role name in $\mathcal{K}$. Variables $X_{T}^3$, denoting concrete role names, and variables $x_a, x_{e_d}$, denoting individuals and datatype constants, respectively, are interpreted in a similar way. From the definitions of $\D$ and $\I$, it follows easily that  $\mathbfcal{M}_{\I,\D}$ satisfies the formulae $\psi_1$-$\psi_{12}$ and $\tau(H)$, for every statement $H \in \mathcal{K}$, and, therefore, that $\mathbfcal{M}_{\I,\D}$ is a model for $\varphi_{\mathcal{K}}$.\qed
%
%
\end{proof}

Some considerations on the expressive power of the logic $\shdlss$ are in order. Despite $\shdlss$ allows one to express existential quantification and at-least number restriction (resp., universal quantification and at-most number restriction) only on the left- (resp., right-) hand side of inclusion axioms, it is more liberal than \sroiqd\space
in the construction of role inclusion axioms since the roles involved are not required to be subject to any ordering relationship. For example, the role hierarchy $\{RS \sqsubseteq S,  RT \sqsubseteq R, VT \sqsubseteq T, VS \sqsubseteq V\}$ presented in \cite{Horrocks2006} and not expressible in \sroiqd\space is admitted by the language of $\shdlss$. Moreover, the notion of simple role is not needed in the definition of role inclusion axioms and of axioms involving number restrictions. In addition, Boolean operators on roles are admitted and can be introduced in inclusion axioms such as, for instance, $R_1 \sqsubseteq R_2 \sqcap R_3$ and $R_1 \sqsubseteq \neg R_2 \sqcup R_3$. Finally, $\shdlss$ treats derived datatypes by admitting datatype terms constructed from data ranges by means of a finite number of applications of the Boolean operators. Basic and derived datatypes can be used inside inclusion axioms involving concrete roles.

\begin{remark}
For a fixed positive integer $h$, a $\shdlss$-knowledge base $\mathcal{K}$ is said to be \emph{$h$-restricted} if an atom of any of the forms $R_1\ldots R_{n_1} \sqsubseteq R$, $\geq_{n_2} \!\!R.C_1 \sqsubseteq C_2$, $\geq_{n_3} \!\!P.t_1 \sqsubseteq t_2$, $C_1 \sqsubseteq {\leq_{n_4} \!\!R. C_2}$, $t_1 \sqsubseteq {\leq_{n_5} \!\!P. t_2}$ occurs in $\mathcal{K}$ only if $n_1,n_2,n_3,n_4,n_5 \leq h$.

It turns out that by using the same function
$\tau$ introduced in the proof of Theorem \ref{theorem} and some additional constraints, the consistency problem for a $h$-restricted $\shdlss$-knowledge base $\mathcal{K}$ can be expressed by a formula $\varphi'_\mathcal{K}$ such that
\begin{itemize}
\item[(i)] $\varphi'_\mathcal{K}$ belongs to the sublanguage $(\flqsr)^h$ of \flqsr, whose satisfiability problem is \textbf{NP}-complete (see \cite{CanNic2013} for details), and

\item[(ii)] the size of $\varphi'_\mathcal{K}$ is polynomially related to that of $\mathcal{K}$.
\end{itemize}
From (i) and (ii) above, and from \textbf{NP}-completeness of the satisfiability problem for propositional logic, it follows immediately that the consistency problem for $h$-restricted $\shdlss$-knowledge bases is \textbf{NP}-complete.

In practice, $h$-restricted $\shdlss$-knowledge bases are quite expressive: for instance, in \cite{santaLAP} we have shown that the ontology \textsf{Ontoceramic}, for ceramics classification, is representable in $(\flqsr)^3$ and, much in the same way, it can be shown that it is representable as a $3$-restricted $\shdlss$-knowledge base.
\end{remark}

\subsection{Translating SWRL-rules into \flqsr-formulae}\label{rules}
The possibility of extending ontologies with rules has become a fundamental requirement to increase the expressiveness and the reasoning power of OWL knowledge bases. In a general sense, a rule is any sentence stating that if a set of premises is satisfied in a given model, then a certain conclusion must be satisfied in the same model. Although OWL is provided with several sorts of conditionals, these are, however, very constrained. Moreover, it is not possible to mix directly classes (concepts) and properties (roles)  and include non-monotonic reasoning such as negation as failure.\footnote{We recall that a logic is non-monotonic if some conclusions can be invalidated when more knowledge is added.}
Such considerations led to the definition of SWRL \cite{swrl}, a rule language combining OWL with the Unary/Binary Datalog fragment of the Rule Markup Language. SWRL allows users to write rules containing OWL constructs providing more reasoning capabilities than OWL alone.

An SWRL-rule $\mathsf{r}$ has the form $(\forall x_1,\ldots,x_n) ( \mathtt{B}  \implies \mathtt{H})$, where:
 \begin{itemize}
 \item[-]  $\mathtt{B}$ (the body of $\mathsf{r}$) and $\mathtt{H}$ (the head of $\mathsf{r}$) are conjunctions of atoms of the following types: $x \in C,~~y \in t,~~\langle x,y \rangle \in R,~~\langle x,y \rangle \in T,~~x=y,~~x \neq y$,
with $C$ a concept name, $t$ a datatype, $R$ an abstract role name, $T$ a concrete role name, and $x$, $y$ either individuals or variables (in the specific cases of atoms of the forms $y \in t$ and $\langle x,y \rangle \in T$, $y$ can be either a datatype constant or a variable), and

\newcommand{\Var}{\textit{Var}}
\item[-]
$\Var(\mathtt{H}) \subseteq \Var(\mathtt{B})=\{x_1,\ldots,x_n\}$, where $\Var(\mathtt{H})$ and $\Var(\mathtt{B})$ are the sets of variables occurring in $\mathtt{H}$ and in $\mathtt{B}$, respectively.
 \end{itemize}

 In Table \ref{SWRL} we give some examples showing how SWRL-rules can be expressed by \flqsr-formulae. For space reasons we do not provide here a formal translation function. However, it is not hard to see that
 it could be constructed by modifying the map $\tau$ introduced in the proof of Theorem \ref{theorem}.
\newcommand{\hasParent}{\textit{hasParent}}
\newcommand{\hasBrother}{\textit{hasBrother}}
\newcommand{\hasUncle}{\textit{hasUncle}}
\newcommand{\Location}{\textit{Location}}
\newcommand{\Trauma}{\textit{Trauma}}
\newcommand{\isLocationOf}{\textit{isLocationOf}}
\newcommand{\isPartOf}{\textit{isPartOf}}
\newcommand{\Person}{\textit{Person}}
\newcommand{\hasAge}{\textit{hasAge}}
\newcommand{\Adult}{\textit{Adult}}
\newcommand{\hasLocation}{\textit{hasLocation}}
\newcommand{\hasRegion}{\textit{hasRegion}}
\newcommand{\Region}{\textit{Region}}
{\small
\begin{longtable}{|c|>{\centering\arraybackslash}m{10cm}|}
\hline
Type of Rule & Rule\\
\hline
SWRL-rule & $ \hasParent(X, Y) , \hasBrother(Y, Z) :- \hasUncle(X, Z).$  \\
\flqsr-rule & $(\forall x)(\forall y)(\forall z)( \langle x,y \rangle \in X^3_{\hasParent} \wedge \langle y,z \rangle \in X^3_{\hasBrother} \rightarrow \langle x,z \rangle \in X^3_{\hasUncle})$\\
\hline
SWRL-rule & $\Location(X),\Trauma(Y),\isLocationOf(X, Y),\isPartOf(X,Z)$  \\
       & :- $\isLocationOf(Z, Y)$ \\
\flqsr-rule & $(\forall x)(\forall y)(\forall z)( x \in X^1_{\Location} \wedge y \in X^1_{\Trauma} \wedge \langle x,z \rangle \in X^3_{\isPartOf} \rightarrow \langle z,y \rangle \in X^3_{\isLocationOf})$\\
\hline
SWRL-rule & $ \Person(X), \hasAge(X, Y), (Y \geq 18) :- \Adult(X)$  \\
\flqsr-rule & $(\forall x)(\forall y)( x \in X^3_{\Person} \wedge \langle x,y \rangle \in X^3_{\hasAge} \wedge y \in X^1_{\geq 18} \rightarrow x \in X^1_{\Adult} )$\\
\hline
SWRL-rule & $ Region(Y), \hasLocation(X, Y) :- \hasRegion(X,Y)$  \\
\flqsr-rule & $(\forall x)(\forall y)( y \in X^3_{\Region} \wedge \langle x,y \rangle \in X^3_{\hasLocation}\rightarrow  \langle x,y \rangle \in X^3_{hasRegion} )$\\
\hline
\caption{ Examples of rule translation.}\label{SWRL}
\end{longtable}}

\section{Conclusions and Future Work}\label{conclusions}
We have introduced the description logic $\shdlss$ which admits, among other features, datatype reasoning, role chain axioms without regularity conditions on roles, min (resp., max) cardinality construct on the left-hand (resp., right-hand) side of inclusion axioms extended to non-simple roles, constructs of full negation, union, and intersection for abstract roles. As discussed at the end of Section \ref{dlss}, the logic $\shdlss$ turns out to be quite expressive, if compared with \sroiqd, the logic underpinning the Web Ontology Language OWL.
However, although $\shdlss$ is endowed with features not supported by \sroiqd, it is not a proper extension of it, as $\shdlss$ admits existential (resp., universal) quantification only on the left-hand (resp., right-hand) side of inclusion axioms.

Through a suitable translation process, we have then shown that the consistency problem for $\shdlss$-knowledge bases can be effectively reduced to the satisfiability problem for the decidable fragment of set theory \flqsr. Moreover, in the restricted case in which a $\shdlss$-knowledge base $\mathcal{K}$ can involve only role chain axioms $R_1 \ldots R_ m \sqsubseteq R$ and inclusion axioms ${\geq_n\!\! R. C_1} \sqsubseteq C_2$, $C_1 \sqsubseteq {\leq_p\!\!R.C_2}$ such that $m$, $n$, and $p$ do not exceed a fixed constant (hence independent of the size of $\mathcal{K}$), we have shown that the consistency problem is \textbf{NP}-complete, as it can be polynomially reduced to the satisfiability problem for a subfragment of \flqsr\ which has an \textbf{NP}-complete decision problem.  Finally, we have also translated SWRL-rules into the \flqsr\ language.


We plan to introduce the constructs of union and intersection of concrete roles and to extend our results to include also datatype groups (here we have considered only a simple form of datatypes) and to admit Boolean operators on concrete roles by defining a suitable strategy of datatype checking. Moreover, we intend to extend the fragment \flqsr\space with metamodelling capabilities \cite{Motik05,GlimmRV10,HomolaKSV14}, so as to make it possible to define concepts containing other concepts and roles (i.e., meta-concepts) and relationships between concepts or between roles (i.e., meta-roles).
Finally, we intend to implement efficient reasoners for suitable fragments of
\flqsr.

\bibliographystyle{ieeetr}
\bibliography{cite} 

\end{document}